\documentclass[10pt]{article}
\usepackage{amssymb}

\newcommand{\la}{\lambda}
\newcommand{\qed}{\hfill$\Box$}
\newtheorem{theorem}{Theorem}
\newtheorem{lemma}{Lemma}
\newtheorem{proposition}{Proposition}
\newtheorem{corollary}{Corollary}
\newtheorem{definition}{Definition}
\newenvironment{proof}{\rm {\bf Proof.}}{}
\begin{document}
\begin{center}
{\bf\Large Decidability of All Minimal Models}\vspace{1.5ex}

{\bf \large(Revised Version -- 2012)}\vspace{1ex}

Vincent Padovani\vspace{1ex}

Equipe Preuves, Programmes et Syst\`emes,\\
Universit\'e Paris VII - Denis Diderot\\
{\tt padovani@pps.univ-paris-diderot.fr}
\end{center}
\begin{quote}
\small {\bf Abstract.} This unpublished note is an alternate, shorter (and hopefully more readable) proof of the decidability of all minimal models. The decidability follows a proof of the existence of a cellular term in each observational equivalence class of a minimal model. 
\end{quote}
The first proof I gave of the decidability of all minimal models \cite{Padovani1995} was far from being easy to understand. I was only an inexperienced student at the time, struggling to solve the problem of the decidability of Higher Order Matching. I was trying to generalise to order five my decidability result at order four (\cite{Padovani1996}, \cite{Padovani2000}) when I realized that every solution of an atomic matching problem (a problem whose right-members are constants of ground type) could be transformed into a {\em cellular} term (the so-called ``transferring'' terms in \cite{Padovani1995}), a term of very simplified structure. The decidability of atomic matching followed immediately from this key-result. A few months later, at the open-problem session of TLCA 1995, Ralph Loader pointed out that another immediate consequence of this decidability result was the existence of a computable selector for the observational equivalence classes of the minimal models of simply-typed lambda-calculus.

Because I was so immersed in the Matching Problem and wanted to prove the decidability of atomic matching at each order, it seemed very natural to prove the existence of cellular representatives by induction on the order of terms. Unfortunately this choice was probably the worst I could make, and resulted in a long, tedious and obfuscated proof. 
Two years later, using the same techniques, Ralph Loader proved the decidability of Unary PCF \cite{Loader1998}. Loader's proof was then drastically simplified by Manfred Schmidt-Schau\ss\ \cite{SchmidtSchauss1999} who gave  a clever, simple and beautiful algorithm to compute a selector for Unary PCF -- {\em a fortiori} for every minimal model. This is the point where I realized that something was probably wrong with my own proof : even if the decidability of all minimal models followed from the existence of cellular representatives at each order,  the latter property was actually independant from the first, and clearly required a proof by induction on the {\em length} of terms.

This unpublished note -- written a few years ago -- presents a short and simple proof of the existence for each term of an observationally equivalent cellular term, followed by a proof of decidability of each minimal model. The proof considers only one ground type, but can be easily extended to finitely many ground types (if you feel it is really necessary, you can try to read \cite{Padovani1995}, or even \cite{Padovani1996} if you can read French). \vfill\pagebreak 

\noindent We consider the $\lambda$-calculus with a single ground type $\circ$, a typing {\em \`a la Church}, and finitely many constants
of type $\circ$. All terms are assumed to be in $\eta$-long form. The notation
 $\lambda y_1\dots y_n.u$ is a shorthand for  $(\lambda y_1\dots (\lambda y_n\,u)\dots)$ and implies
 that $u$ is of ground type. 

Let $t,t'$ be closed terms of the same type.
 We say that $t$ and $t'$ are {\em observationally equivalent} if and only if the following property holds:
\begin{itemize}\parskip 0ex
\item $t,t':\circ$ and $t =_\beta t'$, or
\item $t,t':A\to B$, and for every closed $u:B$, we have $(t\,u)\equiv (t'\,u)$.
\end{itemize}
The following lemma is a well-known result:
\begin{lemma}\mbox{} (Context Lemma)
If $u\equiv u'$ then for all $t$, $(t\,u)\equiv (t'\,u)$.
\end{lemma}
The following property 
is {\em false} in a simply-typed $\la$-calculus dealing with higher-order constants, {\em true} if all constants are of ground type :
\begin{proposition}\label{rep} (Stretching lemma)
Let $t = \lambda y_1\dots y_n.M[u]$ be a closed term, where $M$ is a context with a hole of ground type. Then $t$ is observationally equivalent
 to $\lambda y_1\dots y_n.M[M[u]]$.
\end{proposition}
\begin{proof} Indeed, $M[M[u]] = M[x][M[u]/x]$, and for all closed $\overline{t}$,
 $M[M[u]][\overline{t}/\overline{y}]
 =  M[\overline{t}/\overline{y}][x] [M[u][\overline{t}/\overline{y}]/x]$,
with either
\begin{itemize}\parskip 0ex
\item $M[\overline{t}/\overline{y}][x] =_\beta x$
 and   $M[M[u]][\overline{t}/\overline{y}] =_\beta
 x[[M[u][\overline{t}/\overline{y}]/x] = M[u][\overline{t}/\overline{y}]$, or,
\item
 $M[\overline{t}/\overline{y}][x] =_\beta a$
 and  $M[M[u]][\overline{t}/\overline{y}] =_\beta a [M[u][\overline{t}/\overline {y}]/x] = a$\qed
\end{itemize}
\end{proof}
\subsection*{1. Cells and cellular terms}
From now on,  by $C[\ ]_1\dots [\ ]_K$ we mean a multiple-hole context whose holes are amongst  $[\ ]_1,\dots, [\ ]_K$.
\begin{definition} A {\em cell}
is a context of the form
$$\Sigma[\ ]_1\dots [\ ]_K = (y\, (C_1[\ ]_1\dots [\ ]_K)\ \dots \ (C_p[\ ]_1\dots [\ ]_K))$$ where all $[\ ]_i$ are holes of ground type and each $C_i[\ ]_1\dots [\ ]_K$ is a context with no free variables.
\end{definition}
By definition each cell contains a unique (head) occurrence of a unique free variable.
\begin{definition} A term $t$ is {\em cellular} if and only if it is of 
the following form:
\begin{itemize}\parskip .2ex
\item $\lambda y_1\dots y_n.a$ where $a$ is a constant, or,
\item $\lambda y_1\dots y_n.\Sigma\,[w_1]_1\dots[w_k]_K$
 where:
\begin{itemize}\parskip .2ex
\item  $\Sigma$ is a cell whose free variable is amongst $y_1,\dots, y_n$,
\item each $\lambda y_1\dots y_n.w_j$ is a cellular term.
\end{itemize}
\end{itemize}
A term is called {\em semi-cellular} if it is cellular, or of the form
 $\lambda y_1\dots y_n.y_i\,u_1\dots u_p$ where 
  each  $\lambda y_1\dots \lambda y_n u_i$ is a cellular term.
\end{definition}
Clearly every cellular (resp. semi-cellular) term  is a closed term.
Note that in the definition of a semi-cellular term, $u_i$ is not necessarily of ground type -- {\em e.g.} if $u_i = \la x_1\dots x_p.t$, then $t$ may contain cells with head variables amongst $y_1\dots, y_n, x_1,\dots, x_p$.

  The introduction of cells will simplify our proofs below, but the following alternate definition of a cellular term will probably be easier to grasp. Let $t = \la y_1\dots y_n.u$ where no variable is simultaneously free and bound in $t$. The term $t$ is cellular if it is closed, and if for every subterm $w = (y_i v_1\dots v_n)$ of $u$, the free variables of $w$ are amongst $y_1,\dots, y_n$. In other words, if some variable $z\neq y_1\dots y_n$ is bound in $t$, then no $y_i$ is allowed to occur between $\la z$ and an occurrence of $z$. For instance, the indentity 
$$\la y_1 y_2.y_1(\la z.y_2\,z) : ((\circ\to\circ)\to\circ)\to((\circ\to\circ)\to\circ)$$
is not cellular, whereas
$$\la y_1 y_2.y_1(\la d.y_2(y_1\la z.z))$$
is a cellular term... observationally equivalent to the first. The following proposition is easily proven by induction on the length of $t$:
\begin{proposition}\label{simplcell} Let $\Sigma$ be a cell whose
 free variable is amongst $y_1,\dots,y_n$.
If $\lambda y_1\dots y_n.M[\Sigma[w_1]\dots [w_K]]$ is cellular, 
(resp. semi-cellular) then 
 $\lambda y_1\dots y_n.M[w_k]$ is cellular (resp. semi-cellular).
\end{proposition}
Note that for all $w_1,\dots,w_K$, 
$\Sigma[w_1]\dots [w_K] = \Sigma[x_1]\dots [x_k]
[w_1/x_1\dots w_K/x_K]$ where $x_1,\dots x_k$ are pairwise distinct variables of ground type.
Furthermore, if the free variable of $\Sigma$
 belongs to $\{y_1,\dots,y_n\}$, then for all closed $t_1,\dots, t_n$,
the normal form of $\Sigma[x_1]\dots [x_k][t_1/y_1\dots t_n/y_n]$
is either equal to some $x_i$, or equal to a constant of ground type.
As a consequence,
\begin{proposition}\label{repcell} (Shrinking lemma) Let $\Sigma$ be
 a cell whose free variable is among $y_1,\dots, y_n$. Then:
$$\lambda y_1\dots y_n.M[\Sigma[w_1]\dots 
[w_{k-1}][N[\Sigma[v_1]\dots[v_K]][w_{k+1}]\dots[w_K]]$$
is observationally equivalent to 
$$\lambda y_1\dots y_n.M[\Sigma[w_1]\dots 
[w_{k-1}][N[v_k]][w_{k+1}]\dots  [w_K]]$$
\end{proposition}
(the names ``stretching'', ``shrinking'' were found by Thierry Joly around 1996).\pagebreak
\subsection*{2. Existence of cellular representatives}
\begin{lemma} Every semi-cellular term $t$ is observationally equivalent to 
a cellular term.
\end{lemma}
\begin{proof}
 By induction on the length of $t$.
 Assume $t = \lambda y_1\dots \lambda y_n.u$ is in normal form, 
 with $u = (y_i\, u_1\dots u_p)$.
 Let $M$ be the minimal context such that
\begin{itemize} 
\item $M\neq [\ ]$,
\item $u = M[t^1]\dots [t^K]$
\item each $t^k$ is a term of ground type, of the form
 $(y\,\dots\,)$ with $y\in\{y_1\dots y_n\}$.
\end{itemize}
For each $k$, $t^k$ is a subterm of some $u_j$, and the closure of this
 latter term is cellular. Hence,
 there exists a cell $\Sigma^k$ and terms $w^k_1\dots w^k_{L_k}$ such that: 
$$t^k = \Sigma^k[w^k_1]\dots[w^k_{L_k}]$$ 
Now, for each $(k,l)$ let
$$N^k_l = M[t^1]\dots[t^{k-1}][w^k_l][t^{k+1}]\dots[t^K]$$
By proposition \ref{simplcell}, for all $(k,l)$,
 $\lambda y_1\dots y_n.N^k_l$ is semi-cellular. By induction hypothesis, there exists a term $u^k_l$
 whose closure is cellular and equivalent to the closure of $N^k_l$.
We define $u'$ as
$$ M[\Sigma^1[u^1_1]\dots[u^1_{L_1}]]\dots [\Sigma^K[u^K_1]\dots[u^K_{L_K}]]$$
Clearly, $t'= \lambda y_1\dots y_n.u'$ is cellular. 
 We claim that $t$ and $t'$ are equivalent. 
By proposition \ref{rep}, $t$ is equivalent to
$\lambda y_1\dots y_n.M[\Sigma^1[u]\dots [u]]\dots [\Sigma^K[u]\dots[u]]$. 
For each $k$, $\lambda y_1\dots y_n.\Sigma^k[u]\dots [u]$ is a closed
 term. By proposition \ref{repcell}, this term is equivalent to 
$\lambda y_1\dots y_n.\Sigma^k[N^k_1]\dots[N^k_{L_k}]$, thereby equivalent to\\
$\lambda y_1\dots y_n.\Sigma^k[u^k_1]\dots[u^k_{L_k}]$.
 The conclusion follows from the definition of $u'$.
\qed\end{proof}

\begin{theorem}\label{cellequiv} Every closed 
 term $t$ is observationally equivalent to a cellular term.
\end{theorem}
\begin{proof}
By induction on the length of $t$.
 If $t = \lambda y_1\dots y_n.(y_i\,u_1\dots u_p)$,
 then by induction hypothesis there exist terms $u'_1,\dots,u'_n$
 such that for each $i$, $\lambda y_1\dots \lambda y_n\, u'_i$
 is a cellular term equivalent to  $\lambda y_1\dots \lambda y_n\, u_i$.
 The term $t' = \lambda y_1\dots y_n.(y_i\,u'_1\dots u'_p)$
 is semi-cellular and equivalent to $t$. The conclusion follows 
 from the preceding lemma. 
\qed\end{proof}
\begin{corollary}
For all types $A$, there exists a term $t:A\rightarrow A$ such that for all terms $w:A$, the normal form of
 $(t\,w)$ is a cellular term equivalent to $w$.
\end{corollary}
\begin{proof}
Define $t$ as the cellular equivalent of the $\eta$-long form of $\lambda x\, x$.\qed
\end{proof}
\subsection*{3. Decidability of all minimal models}
For any finite set of constants $\cal C$, we write $\equiv_{\cal C}$ the restriction
 of observational equivalence between closed terms whose constants belong to $\cal C$.
\begin{theorem} There exists a comptable function $\cal R$ such that
 for all types $A$, and for any finite set of constants of ground type
 $\cal C$, ${\cal R}(A,{\cal C})$ is a finite list of terms containing
 a representative of each $\equiv_{\cal C}$-class.
\end{theorem}
\begin{proof}
Let $A = A_1\dots A_n\rightarrow\circ$, where
  $A_i = B^i_1\dots B^i_{p_i}\rightarrow\circ$.
 Let $y_1:A_1,\dots, y_n:A_n$.
For each $i\in[1\dots n]$, let $K_i$ be any integer greater\footnote{If $|{\cal R}(B^i_j,{\cal C})| = k_j$ then we can take $K_i = |{\cal C}|^{(k_1\times\dots\times k_{p_i})}$.}
 than the number of $\equiv_{\cal C}$-classes of type $A_i$.
For each $j\in[1\dots p_i]$,
 let $W^j_i$ be a
 complete set of representatives for the pair
 $(B^i_j,{\cal C}\cup\{d^i_1\dots d^i_{K_i}\})$, where the $d^i_k$'s
  are fresh constants of type $\circ$.
We let $V_i$ be the set of all terms of the form 
$(y_i\,{\overline w})$.
 where ${\overline w}\in  \Pi_{j=1}^{p_i} W^j_i$. We define $\cal R(A,{\cal C})$ as the least set of terms $R$ satisfying:
\begin{itemize}
\item $\lambda \overline y.a\in R$  for all $a\in{\cal C}$.
\item if $\lambda \overline y.w_1,\dots,\lambda \overline y.w_{K_i}\in R$,
 and if $v\in V_i$ is not used in the construction of these terms,
 then $\lambda \overline y.v[w_1/d^i_1]\dots[w_{K_i}/d^i_{K_i}]\in R$.
\end{itemize}
Clearly, $R$ is a finite set of cellular terms.
 Let $t$ be any closed term.  We shall prove that there exists
 a term in $R$ equivalent to $t$. By the preceding theorem, we can assume that $t$ is cellular. We proceed by induction on the length of $t$.

Suppose
 $t =_\alpha \la y_1\dots y_n.\Sigma[s_1]\dots [s_L]$ where $\Sigma$ is a cell of head
 variable $y_i$. The number
 of holes appearing in the normal forms of all $\Sigma[u_i/y_i]$
 where $u_i$ is a closed term, is bounded by $K_i$. If $[\ ]_l$
 does not appear in these normal forms, 
 then we can replace  $s_l$ 
 with an arbitrary constant, yielding a term of same class.
 As a consequence, we can assume that $L\leq K_i$. Then there exists in $V_i$ a term $v_{\Sigma}$ such that
$$\la y_i.\Sigma[d^i_1]\dots [d^i_L]\equiv \la y_i.v_{\Sigma}$$
By induction hypothesis there exists $r_1,\dots, r_L$ such that  $\la\overline y.s_1 \equiv \la\overline y.r_1\in R,\dots,
 \la\overline y.s_L \equiv \la\overline y.r_L\in R$. Then:
$$t \equiv \la y_1\dots y_n.v_\Sigma[r_1/d^i_1\dots r_L/d^i_L]$$
The conclusion follows from the fact that, as in proposition
 \ref{repcell}, we do not need to use $v_\Sigma$
 in the construction of $\la \overline y.r_1,\dots \la \overline y.r_L$,
 that is to say,
$$ \la\overline y.v_\Sigma[r_1/d^i_1
\dots r_{l-1}/d^i_{l-1},
 M[v_\Sigma[r'_1/d^i_1\dots r'_L/d^i_L]]/d^i_l,
r_{l+1}/d^i_{l+1},\dots,  r_L/d^i_L]$$
is equivalent to 
$$\la\overline y.v_\Sigma[r_1/d^i_1
\dots r_{l-1}/d^i_{l-1},
 r'_l/d^i_l,
r_{l+1}/d^i_{l+1},\dots,  r_L/d^i_L]$$
\mbox{}\qed
\end{proof}\pagebreak

\noindent{\em Remark.} Call {\em hereditary cellular} every cellular $t$ such that
 for each cell 
$$y(C_1[\ ]_1\dots [\ ]_K)\,\dots (C_n[\ ]_1\dots [\ ]_K)$$ in $t$, each 
 $\la x_1\dots \la x_K\,C_j[x_1]\dots [x_K]$ is hereditary cellular.
 Note that all terms returned by the algorithm are hereditary cellular. It it not too difficult to prove that all terms returned by a restriction of  Schmidt-Schau\ss' 
algorithm to a minimal model (see \cite{Loader1997} for Loader's presentation of Schmidt-Schau\ss' 
algorithm) are also hereditary cellular.

\end{document}